\documentclass[11pt]{article}

\usepackage[a4paper,margin=1in]{geometry}
\usepackage[T1]{fontenc}
\usepackage[utf8]{inputenc}
\usepackage{lmodern}
\usepackage{amsmath,amssymb,amsthm,mathtools}
\usepackage{booktabs}
\usepackage{enumitem}
\usepackage{cite}
\usepackage{microtype}
\usepackage{xcolor}
\usepackage{comment}
\usepackage[colorlinks=true,citecolor=blue,linkcolor=blue,urlcolor=blue]{hyperref}

\allowdisplaybreaks[3]
\numberwithin{equation}{section}

\AtBeginDocument{%
  \setlength{\abovedisplayskip}{6pt plus 2pt minus 3pt}%
  \setlength{\belowdisplayskip}{6pt plus 2pt minus 3pt}%
  \setlength{\abovedisplayshortskip}{4pt plus 2pt minus 2pt}%
  \setlength{\belowdisplayshortskip}{4pt plus 2pt minus 2pt}%
}

\setlist{itemsep=2pt,topsep=4pt,leftmargin=*}

\makeatletter
\renewcommand\section{\@startsection{section}{1}{\z@}%
  {-1.8ex plus -0.5ex minus -0.2ex}%
  {0.8ex plus 0.2ex}%
  {\normalfont\large\bfseries}}
\renewcommand\subsection{\@startsection{subsection}{2}{\z@}%
  {-1.4ex plus -0.4ex minus -0.2ex}%
  {0.6ex plus 0.2ex}%
  {\normalfont\normalsize\bfseries}}
\makeatother

\theoremstyle{plain}
\newtheorem{theorem}{Theorem}[section]

\newtheorem{lemma}[theorem]{Lemma}

\theoremstyle{definition}
\newtheorem{definition}[theorem]{Definition}

\theoremstyle{remark}
\newtheorem{remark}[theorem]{Remark}

\title{Polynomially Improved Lower Bounds for Trifferent Codes via Locally Sparse $3$-Uniform Hypergraphs%
\thanks{This research was supported by the National Key Research and Development Program of China under Grant 2025YFC3409900, the National Natural Science Foundation of China under Grant 12231014, and the Beijing Scholars Program. \emph{(Corresponding author: Yubo Sun.)}}}

\renewcommand{\thefootnote}{\fnsymbol{footnote}}

\author{
Xuejiao Han\thanks{School of Mathematical Sciences, Capital Normal University, Beijing 100048, China. Email: \texttt{hanxj4425@126.com}.}
\and
Yubo Sun\thanks{Institute of Mathematics and Interdisciplinary Sciences, Xidian University, Xi'an 710126, China. Email: \texttt{ybsun@cnu.edu.cn}.}
\and
Gennian Ge\thanks{School of Mathematical Sciences, Capital Normal University, Beijing 100048, China. Email: \texttt{gnge@zju.edu.cn}.}
}

\date{}

\begin{document}

\maketitle

\setcounter{footnote}{0}
\renewcommand{\thefootnote}{\arabic{footnote}}

\begin{abstract}
A ternary code is  \emph{trifferent} if every three distinct codewords have a coordinate in which their symbols are pairwise distinct. Let $T(n)$ be the maximum size of a trifferent code of length $n$. The classical K\"orner--Marton construction gives $T(n)\ge c_0(9/5)^{n/4}$ for an absolute constant $c_0>0$. We prove the polynomial strengthening
$T(n)\ge c\sqrt{n}(9/5)^{n/4}$ for an absolute constant $c>0$. Our proof refines the outer-code step in the K\"orner--Marton concatenation. We encode non separating triples as edges of a $3$-uniform hypergraph, randomly thin its vertex set, and remove high-degree vertices together with all remaining Berge cycles of lengths two and three. The resulting locally sparse hypergraph admits a large independent set by a theorem of Verstraete and Wilson, producing the additional factor $\sqrt n$. Concatenation with the length-four Tetra code then yields the stated lower bound.
\end{abstract}


\section{Introduction}

Let $[q]:=\{1,\ldots,q\}$ be an alphabet of size $q$, and let $[q]^n$ denote the set of length-$n$ words over the alphabet $[q]$. A code $\mathcal C\subseteq[q]^n$ is called a \emph{$q$-perfect hash code} if, for every $q$ distinct codewords $x_1,\ldots,x_q\in\mathcal C$, there exists a coordinate $i\in[n]$ such that $\{x_1(i),\ldots,x_q(i)\}=[q]$. Equivalently, every $q$-tuple of distinct codewords is separated in at least one coordinate.

Perfect hash codes arise naturally in theoretical computer science \cite{KM88,WX01}, cryptography \cite{BBKT10,Zaverucha10}, and zero-error list decoding \cite{Elias88,BR22}. They are also connected to several apparently unrelated objects in finite geometry and coding theory, including strong blocking sets \cite{BDGP24,ABDN24}, minimal codes \cite{BDGP24,ABNR22}, and covering codes \cite{HN21}.

The case $q=2$ is immediate, since every pair of distinct binary words differs in some coordinate.  The first nontrivial, and by far the most difficult, case is therefore $q=3$.  A perfect $3$-hash code over the ternary alphabet is commonly called a \emph{trifferent code}.  Thus, a code $\mathcal C\subseteq [3]^n$ is trifferent if, for every three distinct
codewords $x,y,z\in\mathcal C$, there exists $i\in[n]$ such that $\{x(i),y(i),z(i)\}=[3]$.  Let $T(n)$ denote the maximum cardinality of a trifferent code of
block length $n$.  Determining $T(n)$, or even its asymptotic order of growth, is known as the \emph{trifference problem}.  Despite its elementary formulation, the problem remains open, and progress toward its resolution has been remarkably slow.

\subsection{Classical Bounds and Recent Progress}

A classical pruning argument gives
\begin{equation}\label{eq:intro-elias}
    T(n)\le 2\left(\frac32\right)^n.
\end{equation}
Indeed, in each coordinate one may delete the least frequent symbol class while retaining at least two thirds of the current code.  After all $n$ coordinates have been processed, at least $(2/3)^n|\mathcal C|$ codewords remain, but at most two symbols occur in each coordinate.  Consequently, no three surviving codewords can triffer, and hence at most two codewords remain.

This argument goes back to the early work of K\"orner \cite{Korner73} and Elias \cite{Elias88}.  Although elementary, it remains the best known general upper bound at the exponential scale. For $q\ge4$, the analogous elementary bounds for perfect hash codes admit exponential improvements \cite{FK84,Arikan94,DGR20,GR22}.  By contrast, improving the corresponding bound in the ternary case has proved notoriously difficult.

The earliest improvements to Equation \eqref{eq:intro-elias} concerned only the leading constant. Fiore, Gnutti, and Polak \cite{DFGP22} determined the exact values of $T(n)$ for $n=5,6$ and used them to improve the general bound.  Kurz \cite{Kur24} subsequently extended the finite-length analysis and obtained a further constant-factor improvement.  More recently, Bhandari and Khetan \cite{BK25} established the first improvement by an unbounded factor:
\begin{equation}\label{eq:intro-bk}
    T(n)\le Cn^{-2/5}\left(\frac32\right)^n
\end{equation}
for some absolute constant $C>0$.  Their argument reduces the problem to constant-weight layers in which each codeword contains a prescribed number of occurrences of one symbol, encodes these layers by suitable auxiliary graphs, and then applies forbidden-subgraph estimates, in particular the K\H{o}v\'ari--S\'os--Tur\'an theorem.  The bound in Equation \eqref{eq:intro-bk} represents a major advance, although it leaves the exponential base $3/2$ unchanged.

The best previously known lower bound is due to K\"orner and Marton \cite{KM88}.  Their construction concatenates a random outer code over an alphabet of size $9$ with a ternary inner code of length $4$, known as the Tetra code, and yields
\begin{equation}\label{eq:intro-km}
    T(n)\ge c_0\left(\frac95\right)^{n/4}
\end{equation}
for some absolute constant $c_0>0$.  Thus, the currently known exponential bases are $ (9/5)^{1/4}\approx 1.1583 $ from below and $3/2$ from above.

Several variants of the trifference problem have also attracted recent attention.  For integers $1\le m\le n$, let $T(n,m)$ denote the maximum size of a ternary code in which every three distinct codewords triffer in at least $m$ coordinates.  Bishnoi, Kielak, Kov\'acs, Nagy, Somlai, Vizer, and Zheng \cite{BKKNSVZ26} obtained bounds for $T(n,m)$ in several parameter ranges and identified a phase transition near $m=2n/9$.  More precisely, for every fixed $\varepsilon>0$, the quantity $T(n,m)$ is bounded independently of $n$ when $m>(2/9+\varepsilon)n$, whereas it grows exponentially with $n$ when $m<(2/9-\varepsilon)n$.  The threshold $2/9$ is natural, since three independent uniformly distributed ternary symbols are pairwise distinct with probability $2/9$.

Another important restriction is obtained by requiring the code to be linear.  Let $T_{\mathrm L}(n)$ be the maximum size of a linear trifferent code. Pohoata and Zakharov \cite{PZ22} proved the first exponential improvement over the pruning bound in this setting: for some absolute $\varepsilon>0$, every sufficiently long linear trifferent code has dimension at most $(1/4-\varepsilon)n$. Bishnoi, D'haeseleer, Gijswijt, and Potukuchi \cite{BDGP24} subsequently revealed a useful geometric structure behind the linear problem.  They proved that a ternary linear code is trifferent if and only if it is a minimal code; through generator matrices, this is also equivalent to a strong blocking set in a finite projective space.  By combining this equivalence with coding-theoretic bounds, they obtained, for large $n$,
\begin{equation}\label{eq:intro-linear}
    \frac13\left(\frac95\right)^{n/4}
    \le T_{\mathrm L}(n)
    \le 3^{n/4.55}.
\end{equation}
In particular, up to a constant factor, the exponential lower bound of K\"orner and Marton can already be achieved within the class of linear codes. Fiore and Dalai \cite{DFD24} later gave a streamlined coding-theoretic treatment and extended the method to linear $q$-ary $k$-hash codes for general $q\ge k\ge3$.

\subsection{Main Results}

In contrast to the recent advances in upper bounds and linear variants, no improvement of the K\"orner--Marton lower bound in Equation \eqref{eq:intro-km} by an unbounded multiplicative factor was previously known.  Our main result provides the first polynomial strengthening of this bound.

\begin{theorem}\label{thm:main}
There exists an absolute constant $c>0$ such that, for all sufficiently large $n$,
\[
    T(n)
    \ge
    c\sqrt{n}\left(\frac95\right)^{n/4}.
\]
\end{theorem}

The proof retains the Tetra code as the inner code and improves the outer-code component of the K\"orner--Marton construction.  To state the corresponding result for outer codes, we call a code $\mathcal D\subseteq[M]^N$ an \emph{$M$-ary perfect $3$-hash code} if, for every three distinct codewords in $\mathcal D$, there exists a coordinate in which their symbols are pairwise distinct. We denote by $H_{M,3}(N)$ the maximum size of an $M$-ary perfect $3$-hash code of length $N$.  Thus, $T(n)=H_{3,3}(n)$.

The Tetra code is the ternary code
\[
    \mathcal T
    :=
    \bigl\{
        (a,b,a+b,a-b):
        a,b\in [3]
    \bigr\}
    \subseteq [3]^4.
\]
It has cardinality $9$ and is itself trifferent \cite[Th.~1]{KM88}.  Consequently, the standard concatenation argument gives
\begin{equation}\label{eq:intro-concat}
    T(4N)\ge H_{9,3}(N)
\end{equation}
for every positive integer $N$.  The main technical result of this paper is the following lower bound on the outer code.

\begin{theorem}\label{thm:outer-main}
There exists an absolute constant $c_1>0$ such that, for all
sufficiently large $N$,
\[
    H_{9,3}(N)
    \ge
    c_1\sqrt{N}\left(\frac95\right)^N.
\]
\end{theorem}

Theorem~\ref{thm:main} follows immediately from Theorem~\ref{thm:outer-main} and Equation \eqref{eq:intro-concat} when $4\mid n$.  The result for arbitrary sufficiently large $n$ is obtained by taking $N=\lfloor n/4\rfloor$ and appending fixed coordinates to the resulting ternary code.

\subsection{Organization}

The remainder of the paper is organized as follows. Section~\ref{sec:hypergraph} introduces the bad-triple hypergraph and the associated independent-set formulation, and outlines our proof strategy.  Section~\ref{sec:counts} establishes the degree and short-cycle estimates.  Section~\ref{sec:thinning} performs the random thinning and local sparsification.
Section~\ref{sec:proof-main} proves Theorem~\ref{thm:outer-main} and derives Theorem~\ref{thm:main}.  Finally, Section~\ref{sec:conclusion} summarizes the results and discusses
directions for further work.

\section{The Bad-Triple Hypergraph and Proof Strategy}
\label{sec:hypergraph}

In this section, we introduce the hypergraph formulation of the outer-code problem and outline the proof of Theorem~\ref{thm:outer-main}.  We begin with the basic terminology and
the independent-set estimate that will be used later.

\subsection{Hypergraph Preliminaries}

A \emph{hypergraph} is a pair $\mathcal H=(V(\mathcal H),E(\mathcal H))$, where $V(\mathcal H)$ is the \emph{vertex set} and $E(\mathcal H)\subseteq 2^{V(\mathcal H)}$ is the \emph{edge set}.  The hypergraph $\mathcal H$ is called \emph{$r$-uniform} if every edge of $\mathcal H$ has cardinality $r$.

For a vertex $v\in V(\mathcal H)$, its \emph{degree} in $\mathcal H$ is 
\[
    d_{\mathcal H}(v):=|\{e\in E(\mathcal H):v\in e\}|.
\]
The \emph{maximum degree} of $\mathcal H$ is 
\[
    \Delta(\mathcal H):=\max_{v\in V(\mathcal H)}d_{\mathcal H}(v).
\]
A set $I\subseteq V(\mathcal H)$ is called an \emph{independent set} in $\mathcal H$ if it contains no edge of
$\mathcal H$; equivalently, $e\nsubseteq I$ for every $e\in E$. The \emph{independence number} of $\mathcal H$, denoted by $\alpha(\mathcal H)$, is the maximum cardinality of an independent set:
\[
    \alpha(\mathcal H)
    :=
    \max\left\{
        |I|:
        I\subseteq V
        \text{ is independent in }\mathcal H
    \right\}.
\]

The independent-set theorem used in our argument requires a locally sparse condition.  We formulate this condition in terms of Berge cycles.

\begin{definition}\label{def:berge}
Let $j\ge2$.  A \emph{Berge $j$-cycle} in a hypergraph consists of pairwise distinct edges $e_1,\ldots,e_j$ and pairwise distinct vertices $v_1,\ldots,v_j$ such that $v_i\in e_i\cap e_{i+1}$ for every $i\in[j]$, where the indices are interpreted cyclically, so that $e_{j+1}=e_1$. A hypergraph is called \emph{locally sparse} if it contains no Berge
cycles of length $2$ or $3$.
\end{definition}

We shall use the following $3$-uniform specialization of an independent-set theorem due to Verstraete and Wilson \cite{VW26}.

\begin{theorem}[Verstraete--Wilson {\cite{VW26}}]
\label{thm:local-sparse}
There exists an absolute constant $c_{\mathrm{VW}}>0$ such that the following holds.  Let $\mathcal G$ be a locally sparse $3$-uniform hypergraph on $M$ vertices, and suppose that $\Delta(\mathcal G)\le \Delta$ for some $\Delta\ge2$.  Then
\[
    \alpha(\mathcal G)
    \ge
    c_{\mathrm{VW}}\,M
    \sqrt{\frac{\log\Delta}{\Delta}}.
\]
\end{theorem}

\subsection{The Bad-Triple Hypergraph}

Having introduced the required hypergraph terminology, we now encode the obstruction to the perfect $3$-hash property as the edge set of a $3$-uniform hypergraph.

A set $\{x,y,z\}\subseteq[9]^N$ of three distinct words is called \emph{bad} if, for every $i\in[N]$, the symbols $x(i),y(i),z(i)$ are not pairwise distinct. In other words, it holds that 
\[
    |\{x(i),y(i),z(i)\}|\leq 2.
\]
Let $\mathcal H_N$ be the $3$-uniform hypergraph with vertex set $V(\mathcal H_N)=V_N:=[9]^N$ and edge set
\[
    E(\mathcal H_N)
    :=
    \left\{
        \{x,y,z\}\in\binom{V_N}{3}:
        \{x,y,z\}\text{ is bad}
    \right\}.
\]

This construction translates the outer-code problem directly into an independent-set problem.  Indeed, a subset $\mathcal D\subseteq V_N$ is independent in $\mathcal H_N$ if and only if it contains no bad triple.  Equivalently, every three distinct words in $\mathcal D$ have pairwise distinct symbols in at least one coordinate. Thus, the independent sets of $\mathcal H_N$ are precisely the $9$-ary perfect $3$-hash codes of length $N$. This implies that
\begin{equation}\label{eq:outer-alpha}
    H_{9,3}(N)=\alpha(\mathcal H_N).
\end{equation}
It therefore suffices to find a large independent set in $\mathcal H_N$.  The next subsection describes how we obtain such a set after passing to a suitably chosen locally sparse induced subhypergraph.

\subsection{Proof Strategy}

The hypergraph $\mathcal H_N$ is too dense for a direct application of Theorem~\ref{thm:local-sparse}.  We therefore retain each vertex independently with a suitably chosen exponentially small probability and consider the induced subhypergraph on the retained vertices. Coordinatewise counting will show that this random subhypergraph still contains exponentially many vertices, while its degree scale is substantially reduced.

We then delete a small exceptional set of vertices in order to control the maximum degree and eliminate all Berge cycles of lengths $2$ and $3$.  The resulting induced subhypergraph $\mathcal G$ remains exponentially large and is locally sparse.  We may therefore apply Theorem~\ref{thm:local-sparse} to obtain
\[
    \alpha(\mathcal G)
    \ge
    c_{\mathrm{VW}}\,|V(\mathcal G)|
    \sqrt{
        \frac{\log\Delta(\mathcal G)}
             {\Delta(\mathcal G)}
    }.
\]

The thinning probability is chosen so that its exponential contribution cancels between the vertex-count and maximum-degree terms.  Moreover, the estimates proved below will show that $\log\Delta(\mathcal G)=\Theta(N)$.  The logarithmic factor in the preceding independent-set bound therefore contributes an additional factor of order $\sqrt N$.
After substituting the resulting estimates for $|V(\mathcal G)|$ and $\Delta(\mathcal G)$, we will obtain an absolute constant $c_1>0$ such that
\[
    \alpha(\mathcal G)
    \ge
    c_1\sqrt N\left(\frac95\right)^N.
\]

Finally, since $\mathcal G$ is an induced subhypergraph of $\mathcal H_N$, every independent set in $\mathcal G$ is also independent in $\mathcal H_N$. Hence,
\[
    H_{9,3}(N)
    =
    \alpha(\mathcal H_N)
    \ge
    \alpha(\mathcal G)
    \ge
    c_1\sqrt N\left(\frac95\right)^N,
\]
which establishes Theorem~\ref{thm:outer-main}.

\section{Degree and Short-Cycle Estimates}
\label{sec:counts}

We now establish the degree and short-cycle estimates needed for the random sparsification argument.  The badness condition is imposed coordinatewise, and hence each relevant configuration count factors over the $N$ coordinates.  It therefore suffices to determine the corresponding one-coordinate count and then raise it to the $N$-th power.

\subsection{Maximum Degree}

We begin with the maximum degree of the bad-triple hypergraph.

\begin{lemma}\label{lem:degree}
Every vertex of $\mathcal H_N$ has degree at most $25^N$.  In particular,
\[
    \Delta(\mathcal H_N)\le 25^N.
\]
\end{lemma}
\begin{proof}
Fix $x\in V_N$. We count ordered pairs $(y,z)\in V_N^2$ such that, in each coordinate $i\in [N]$, the symbols $x(i),y(i),z(i)$ are not pairwise distinct. This count includes pairs for which the three vertices are not distinct and counts each edge containing $x$ twice; hence it is an upper bound for the degree of $x$.

In a fixed coordinate $i\in [N]$, once $x(i)$ is fixed, there are $9^2$ possible choices for the ordered pair $(y(i),z(i))$. The forbidden choices are precisely those for which $x(i),y(i),z(i)$ are pairwise distinct, and there are $8\cdot 7$ such choices. Hence the number of admissible choices in one coordinate is $9^2-8\cdot7=25$.  Since the coordinate constraints are independent, the total number of admissible ordered pairs $(y,z)$ is $25^N$.  Therefore, $d_{\mathcal H_N}(x)\le 25^N$.  Since $x$ was arbitrary, the result follows.
\end{proof}

\begin{remark}
The preceding argument in fact gives the slightly stronger estimate $d_{\mathcal H_N}(x)\le 25^N/2$, because every edge containing $x$ contributes two ordered pairs. The coarser bound in Lemma~\ref{lem:degree} is sufficient for our purposes.
\end{remark}

\subsection{Short Berge Cycles}

We next estimate the number of Berge cycles of lengths $2$ and $3$.  Recall that a Berge $2$-cycle consists of two distinct edges sharing at least two vertices.  Consequently, once all Berge $2$-cycles have been removed from a $3$-uniform hypergraph, any two remaining edges intersect in at most one vertex. It follows that every Berge $3$-cycle in the resulting hypergraph spans exactly six vertices.  Indeed, its three edges contain three distinct intersection vertices and one additional vertex in each edge; the absence of Berge $2$-cycles forces these six vertices to be distinct.  It is therefore enough to count Berge $2$-cycles and six-vertex Berge $3$-cycles.

\begin{lemma}\label{lem:short-cycles}
For $j\in\{2,3\}$, define $K_j:=25^j+8\cdot7^j$. Specifically, 
\begin{align*}
    K_2 &=25^2+8\cdot7^2=9\cdot113,\\
    K_3 &=25^3+8\cdot7^3=18369.
\end{align*}
Then the following statements hold:
\begin{enumerate}[label=\textup{(\roman*)}]
    \item $\mathcal H_N$ contains at most $K_2^N= (9\cdot 113)^N$ Berge $2$-cycles;
    \item $\mathcal H_N$ contains at most $K_3^N= 18369^N$ Berge $3$-cycles whose three edges span exactly six vertices.
\end{enumerate}
\end{lemma}

\begin{proof}
We count labelled configurations, allowing repeated vertices or edges. This can only increase the count and therefore gives an upper bound on the number of cycles.

For $j\in\{2,3\}$, we write such a labelled configuration as $e_i=\{v_i,v_{i+1},u_i\}$ for $i\in[j]$, where the indices are interpreted cyclically, so that $v_{j+1}=v_1$.  When $j=2$, this represents two edges with two common vertices.  When $j=3$, it represents a Berge $3$-cycle whose edges span six vertices.

Fix a coordinate $t\in[N]$, and write $a_i:=v_i(t)$ and $b_i:=u_i(t)$.  For $e_i$ to be bad in coordinate $t$, the three symbols $a_i,a_{i+1},b_i$ must not be pairwise distinct.  Once $a_i$ and $a_{i+1}$ are fixed, the number of admissible choices for $b_i$ is
\[
    M_{a_i,a_{i+1}}
    :=
    \begin{cases}
        9, & a_i=a_{i+1},\\
        2, & a_i\ne a_{i+1}.
    \end{cases}
\]
Indeed, if $a_i=a_{i+1}$, then every value of $b_i$ is admissible; if $a_i\ne a_{i+1}$, then $b_i$ must equal one of these two symbols.

Let $M=(M_{a,a'})_{a,a'\in[9]}$.  If $I$ and $J$ denote the $9\times9$ identity matrix and the $9\times9$ all-one matrix, respectively, then $M=7I+2J$.  Therefore, the number of admissible symbol assignments in coordinate $t$ is
\[
\begin{aligned}
    \sum_{a_1,\ldots,a_j\in[9]}
        M_{a_1,a_2}
        M_{a_2,a_3}
        \cdots
        M_{a_j,a_1}
        =\operatorname{tr}(M^j).
\end{aligned}
\]
Observe that the eigenvalues of $M$ are $25$ with multiplicity one and $7$ with multiplicity eight, then $\operatorname{tr}(M^j)=25^j+8\cdot7^j=K_j$.

Finally, the constraints in different coordinates are independent. Thus, the number of labelled configurations over all $N$ coordinates is at most $K_j^N$.  Taking $j=2$ and $j=3$ proves assertions \textup{(i)} and \textup{(ii)}, respectively.
\end{proof}

\section{Random Thinning and Local Sparsification}\label{sec:thinning}

We next extract a large locally sparse induced subhypergraph of $\mathcal H_N$. Fix $\beta=0.205$ and define
\begin{equation}\label{eq:parameters}
\begin{aligned}
        p&=\beta^N,\\
        \mu&=9^Np=(9\beta)^N,\\
        D&=25^Np^2=(25\beta^2)^N.
\end{aligned}
\end{equation}
Here $p$ is the vertex-retention probability, $\mu$ is the expected number of retained vertices, and $D$ is the natural degree scale after thinning.  The precise value of $\beta$ is immaterial; we only use
\begin{equation}\label{eq:beta-conditions}
        25\beta^2>1,
        \qquad
        113\beta^3<1,
        \qquad
        \frac{18369}{9}\beta^5<1.
\end{equation}
The first inequality ensures that $D$ grows exponentially with $N$, while the other two make the expected numbers of relevant short cycles negligible compared with $\mu$.

\begin{lemma}\label{lem:locally-sparse-subgraph}
For all sufficiently large $N$, there is a subset $R\subseteq V_N$ such that the induced hypergraph $\mathcal G=\mathcal H_N[R]$ satisfies
\begin{enumerate}[label=\textup{(\roman*)}]
    \item $|R|\ge\mu/4$;
    \item $\Delta(\mathcal G)\le16D$;
    \item $\mathcal G$ is locally sparse.
\end{enumerate}
\end{lemma}

Before proceeding to the proof, we recall two standard probabilistic tools: the Chernoff bound and Markov's inequality.

\begin{lemma}[Chernoff bound]\label{lem:chernoff}
Let $X=\sum_{i=1}^m X_i$, where $X_1,\ldots,X_m$ are independent Bernoulli random variables, and let $\mathbb E X=\mu$. Then, for every $0<\delta<1$,
\[
        \mathbb P\bigl(X\le (1-\delta)\mu\bigr)
        \le
        \exp\!\left(-\frac{\delta^2\mu}{2}\right).
\]
\end{lemma}

\begin{lemma}[Markov's inequality]\label{lem:markov}
Let $Y$ be a nonnegative random variable. Then, for every $a>0$,
\[
        \mathbb P(Y\ge a)\le \frac{\mathbb E Y}{a}.
\]
\end{lemma}

\begin{proof}[Proof of Lemma \ref{lem:locally-sparse-subgraph}]
Choose a random subset $S\subseteq V_N$ by retaining each vertex independently with probability $p$, and consider the induced hypergraph $\mathcal H_N[S]$.   Let $X=|S|$.  Then $\mathbb E X=\mu$, and the Chernoff bound gives
\begin{equation}\label{eq:chernoff-S}
        \mathbb P(X<\mu/2)
        \le \exp(-\mu/8)
        =o(1).
\end{equation}

Let $Z_2$ be the number of Berge $2$-cycles in $\mathcal H_N[S]$, and let $Z_3$ be the number of Berge $3$-cycles in $\mathcal H_N[S]$ whose edges span exactly
six vertices.  For $j\in\{2,3\}$, it follows by Lemma~\ref{lem:short-cycles} that
\begin{equation}\label{eq:cycle-expectation}
\begin{aligned}
        \mathbb E Z_j
        \le p^{2j}K_j^N=\mu
          \left(\frac{K_j}{9}\beta^{2j-1}\right)^N.
\end{aligned}
\end{equation}
Since $K_2=9\cdot113$, the second inequality in \eqref{eq:beta-conditions} gives
\[
        \mathbb E Z_2
        \le \mu(113\beta^3)^N
        =o(\mu).
\]
Similarly, $K_3=18369$, and the third inequality in \eqref{eq:beta-conditions} gives
\[
        \mathbb E Z_3
        \le
        \mu\left(\frac{18369}{9}\beta^5\right)^N
        =o(\mu).
\]

It remains to control vertices of abnormally large degree.  Fix a vertex $x\in V_N$.  Conditional on $x\in S$, every edge of $\mathcal H_N$ containing $x$ survives with probability $p^2$. It then follows by  Lemma~\ref{lem:degree} that
\[
        \mathbb E\bigl[d_{\mathcal H_N[S]}(x)\mid x\in S\bigr]
        \le p^2 25^N
        =D.
\]
By Markov's inequality,
\[
\begin{aligned}
        \mathbb P\bigl(x\in S
        \text{ and }d_{\mathcal H_N[S]}(x)>16D\bigr)
        \le
        p\,\frac{D}{16D}
        =\frac{p}{16}.
\end{aligned}
\]
Let $Y$ denote the number of vertices $x\in S$ with $d_{\mathcal H_N[S]}(x)>16D$.  Summing over all $x\in V_N$, we obtain
\begin{equation}\label{eq:high-degree-expectation}
        \mathbb E Y\le\frac{\mu}{16}.
\end{equation}

The estimates above imply that, for all sufficiently large $N$,
\[
        \mathbb E(Y+Z_2+Z_3)\le\frac{\mu}{8}.
\]
Further applying Markov's inequality gives
\[
        \mathbb P(Y+Z_2+Z_3>\mu/4)\le\frac12.
\]
Together with \eqref{eq:chernoff-S}, this shows that, with positive probability,
\begin{equation}\label{eq:good-realization}
        |S|\ge\frac{\mu}{2},
        \qquad
        Y+Z_2+Z_3\le\frac{\mu}{4}.
\end{equation}

Fix a realization of $S$ satisfying \eqref{eq:good-realization}. First remove every vertex from $S$ whose degree in $\mathcal H_N[S]$ exceeds $16D$.  From the remaining set, remove one vertex from each Berge $2$-cycle and from each six-vertex Berge $3$-cycle.  Let $R$ be the set of vertices that remain.  At most $Y+Z_2+Z_3$ vertices are deleted, so
\[
        |R|
        \ge |S|-(Y+Z_2+Z_3)
        \ge\frac{\mu}{4}.
\]
Moreover, $\Delta(\mathcal H_N[R])\le 16D$, since $\mathcal H_N[R]$ is an induced subhypergraph of $\mathcal H_N[S]$ and all vertices of degree greater than $16D$ in $\mathcal H_N[S]$ were removed.

It remains to verify that $\mathcal H_N[R]$ is locally sparse.  By construction, $\mathcal H_N[R]$ contains no Berge $2$-cycle and no Berge $3$-cycle whose three edges span exactly six vertices. Suppose that $\mathcal H_N[R]$ contains a Berge $3$-cycle. Since there is no Berge $2$-cycle, any two edges of this Berge $3$-cycle intersect in exactly one vertex. In a $3$-uniform hypergraph, this forces the three edges to span exactly six vertices, which leads to a contradiction. Therefore, $\mathcal H_N[R]$ is locally sparse and the proof is completed.
\end{proof}

\section{Proof of The Main Lower Bound}\label{sec:proof-main}

We now combine Lemma~\ref{lem:locally-sparse-subgraph} with Theorem~\ref{thm:local-sparse} to prove Theorem \ref{thm:outer-main}.

\begin{proof}[Proof of Theorem~\ref{thm:outer-main}]
Let $R\subseteq V_N$ be supplied by Lemma~\ref{lem:locally-sparse-subgraph}, and set $\mathcal G=\mathcal H_N[R]$ and $M=|R|$.  Then
\[
        M\ge\frac{(9\beta)^N}{4},
        \qquad
        \Delta(\mathcal G)\le16(25\beta^2)^N.
\]
Because $25\beta^2>1$, the quantity $D=(25\beta^2)^N$ tends to infinity. Set $\Delta_0=16D$.  For all sufficiently large $N$, we have $\Delta_0\ge2$.  By Lemma \ref{lem:locally-sparse-subgraph}, $\mathcal G$ is locally sparse. Then Theorem~\ref{thm:local-sparse} gives
\[
        \alpha(\mathcal G)
        \ge
        c_{\mathrm{VW}}M
        \sqrt{\frac{\log\Delta_0}{\Delta_0}}.
\]
Since $\log\Delta_0=N\log(25\beta^2)+O(1)=\Theta(N)$, there is an absolute constant $c_1>0$ such that
\begin{align*}
        \alpha(\mathcal G)
        &\ge
        c_1(9\beta)^N
        \sqrt{\frac{N}{(25\beta^2)^N}}\\
        &=c_1\sqrt N
          \left(\frac{9\beta}{5\beta}\right)^N\\
        &=c_1\sqrt N\left(\frac95\right)^N.
\end{align*}
Since $\mathcal G$ is an induced subhypergraph of $\mathcal H_N$, every independent set in $\mathcal G$ is also independent in $\mathcal H_N$.  Equation \eqref{eq:outer-alpha} now gives
\[
        H_{9,3}(N)
        =\alpha(\mathcal H_N)
        \ge\alpha(\mathcal G)
        \ge c_1\sqrt N\left(\frac95\right)^N,
\]
which proves Theorem~\ref{thm:outer-main}.
\end{proof}

\begin{proof}[Proof of Theorem~\ref{thm:main}]
For every sufficiently large $N$, Equation~\eqref{eq:intro-concat} and Theorem~\ref{thm:outer-main} yield
\[
        T(4N)
        \ge H_{9,3}(N)
        \ge c_1\sqrt N\left(\frac95\right)^N.
\]
Now let $n$ be arbitrary and set $N=\lfloor n/4\rfloor$.  Append $n-4N$ fixed coordinates to each word of a trifferent code of length $4N$.  This operation preserves trifference, and hence $T(n)\ge T(4N)$. For all sufficiently large $n$, we have $N\ge n/8$ and $0\le n-4N\le3$.  Consequently,
\[
\begin{aligned}
        \sqrt N\left(\frac95\right)^N
        &\ge
        \frac{1}{\sqrt8}
        \left(\frac95\right)^{-3/4}
        \sqrt n\left(\frac95\right)^{n/4}.
\end{aligned}
\]
Absorbing the fixed factor into the constant proves that
\[
        T(n)
        \ge c\sqrt n\left(\frac95\right)^{n/4}
\]
for all sufficiently large $n$.
\end{proof}

\section{Conclusion and Future Directions}
\label{sec:conclusion}

We proved the first unbounded-factor improvement on the classical K\"orner--Marton lower bound for trifferent codes. Specifically, for all sufficiently large $n$, 
\[
    T(n)\ge c\sqrt n\left(\frac95\right)^{n/4}.
\] 
Our proof preserves the length-four Tetra inner code and improves the outer-code construction. We model bad triples by a $3$-uniform
hypergraph, apply random thinning, and remove high-degree vertices and short Berge cycles. The resulting locally sparse hypergraph admits a large independent set by the theorem of Verstraete and Wilson. The thinning parameter cancels from the exponential term, while the logarithmic factor in the independence bound produces the additional factor $\sqrt n$.

The same framework extends naturally to generalized $m$-trifferent codes. Fix an integer $m\ge 1$ and an outer block length $N$. We define a $3$-uniform hypergraph $\mathcal H_{N,m}$ on the vertex set $[9]^N$ by declaring three distinct words $x,y,z\in[9]^N$ to form an edge whenever they are separated in at most $m-1$ coordinates. Thus, the independent sets of $\mathcal H_{N,m}$ are exactly the $9$-ary codes in which every three distinct codewords are separated in at least $m$ coordinates. For each fixed $m$, the degree and short-cycle counts have the same exponential growth rates as those obtained in Section~\ref{sec:counts}, with only additional polynomial factors in $N$.
Consequently, the random thinning and local sparsification arguments apply with only minor modifications and yield a constant $c_m>0$, depending only on $m$, such that
\[
    \alpha(\mathcal H_{N,m})
    \ge
    c_m N^{1-m/2}
    \left(\frac95\right)^N
\]
for all sufficiently large $N$. Concatenating an independent set in $\mathcal H_{N,m}$ with the length-four Tetra code then produces a ternary $m$-trifferent code of length $4N$, and hence
\[
    T(4N,m)
    \ge
    c_m N^{1-m/2}
    \left(\frac95\right)^N.
\]
Finally, taking $N=\lfloor n/4\rfloor$ and appending at most three fixed coordinates, which preserves the $m$-trifference property, gives a constant $c_m'>0$ such that
\[
    T(n,m)
    \ge
    c_m' n^{1-m/2}
    \left(\frac95\right)^{n/4}
\]
for all sufficiently large $n$.

A natural question is whether the polynomial factor in our lower bound can be further improved. An even more challenging direction is to improve the exponential base, either in the lower bound or in the corresponding upper bound. It would also be interesting to investigate generalized $m$-trifferent codes and $q$-perfect hash codes. Finally, since our argument is nonconstructive, obtaining an explicit construction with a comparable polynomial improvement remains an important open problem.


\begin{thebibliography}{99}

\bibitem{ABNR22}
G.~N.~Alfarano, M.~Borello, A.~Neri, and A.~Ravagnani, ``Three combinatorial perspectives on minimal codes,'' \emph{SIAM Journal on Discrete Mathematics}, vol.~36, no.~1, pp.~461--489, 2022.

\bibitem{Arikan94}
E.~Ar{\i}kan, ``An upper bound on the zero-error list-coding capacity,'' \emph{IEEE Transactions on Information Theory}, vol.~40, no.~4, pp.~1237--1240, Jul.~1994.


\bibitem{ABDN24}
N.~Alon, A.~Bishnoi, S.~Das, and A.~Neri, ``Strong blocking sets and minimal codes from expander graphs,'' \emph{Transactions of the American Mathematical Society}, vol.~377, no.~8, pp.~5389--5410, 2024.

\bibitem{BBKT10}
A.~Barg, G.~R.~Blakley, G.~Kabatiansky, and C.~Tavernier, ``Robust parent-identifying codes,'' in \emph{2010 IEEE Information Theory Workshop}, 2010, pp.~1--4.

\bibitem{BK25}
S.~Bhandari and A.~Khetan, ``Improved upper bound for the size of a trifferent code,'' \emph{Combinatorica}, vol.~45, Art.~no.~2, 2025.

\bibitem{BR22}
S.~Bhandari and J.~Radhakrishnan, ``Bounds on the zero-error list-decoding capacity of the $q/(q-1)$ channel,'' \emph{IEEE Transactions on Information Theory}, vol.~68, no.~1, pp.~238--247, 2022.

\bibitem{BDGP24}
A.~Bishnoi, J.~D'haeseleer, D.~Gijswijt, and A.~Potukuchi, ``Blocking sets, minimal codes and trifferent codes,'' \emph{Journal of the London Mathematical Society}, vol.~109, no.~6, Art.~no.~e12938, 2024.

\bibitem{BKKNSVZ26}
A.~Bishnoi, B.~Kielak, B.~Kov\'acs, Z.~L.~Nagy, G.~Somlai, M.~Vizer, and Z.~Zheng, ``The generalized trifference problem,'' \emph{IEEE Transactions on Information Theory}, vol.~72, no.~5, pp.~2907--2914, 2026

\bibitem{DGR20}
M.~Dalai, V.~Guruswami, and J.~Radhakrishnan, ``An improved bound on the zero-error list-decoding capacity of the
$4/3$ channel,'' \emph{IEEE Transactions on Information Theory}, vol.~66, no.~2, pp.~749--756, 2020.

\bibitem{DFD24}
S.~D.~Fiore and M.~Dalai, ``Upper bounds on the rate of linear $q$-ary $k$-hash codes,'' in \emph{2024 IEEE International Symposium on Information Theory}, 2024, pp.~2610--2615.

\bibitem{DFGP22}
S.~D.~Fiore, A.~Gnutti, and S.~Polak, ``The maximum cardinality of trifferent codes with lengths 5 and 6,'' \emph{Examples and Counterexamples}, vol.~2, Art.~no.~100051, 2022.

\bibitem{Elias88}
P.~Elias, ``Zero error capacity under list decoding,'' \emph{IEEE Transactions on Information Theory}, vol.~34, no.~5, pp.~1070--1074, 1988.

\bibitem{FK84}
M.~L. Fredman and J.~Koml{\'o}s, ``On the size of separating systems and families of perfect hash functions,'' \emph{SIAM Journal on Algebraic and Discrete Methods}, vol.~5, no.~1, pp.~61--68, 1984.

\bibitem{GR22}
V.~Guruswami and A.~Riazanov, ``Beating Fredman--Koml{\'o}s for perfect $k$-hashing,'' \emph{Journal of Combinatorial Theory, Series A}, vol.~188, Art.~no.~105580, 2022.

\bibitem{HN21}
T.~H{\'e}ger and Z.~L.~Nagy, ``Short minimal codes and covering codes via strong blocking sets in projective spaces,'' \emph{IEEE Transactions on Information Theory}, vol.~68, no.~2, pp.~881--890, 2021.

\bibitem{Korner73}
J.~K{\"o}rner, ``Coding of an information source having ambiguous alphabet and the entropy of graphs,'' in \emph{6th Prague Conference on Information Theory}, 1973, pp.~411--425.

\bibitem{KM88}
J.~K{\"o}rner and K.~Marton, ``New bounds for perfect hashing via information theory,'' \emph{European Journal of Combinatorics}, vol.~9, no.~6, pp.~523--530, 1988.

\bibitem{Kur24}
S.~Kurz, ``Trifferent codes with small lengths,'' \emph{Examples and Counterexamples}, vol.~5, Art.~no.~100139, 2024.

\bibitem{PZ22}
C.~Pohoata and D.~Zakharov, ``On the trifference problem for linear codes,'' \emph{IEEE Transactions on Information Theory}, vol.~68, no.~11, pp.~7096--7099, Nov.~2022.

\bibitem{VW26}
J.~Verstraete and C.~Wilson, ``Independent sets in hypergraphs,'' \emph{Random Structures \& Algorithms}, vol.~68, no.~1, Art.~no.~e70047, 2026.

\bibitem{WX01}
H.~Wang and C.~Xing, ``Explicit constructions of perfect hash families from algebraic curves over finite fields,'' \emph{Journal of Combinatorial Theory, Series A}, vol.~93, no.~1, pp.~112--124, 2001.

\bibitem{Zaverucha10}
G.~Zaverucha, ``Hash families and cover-free families with cryptographic applications,'' Ph.D. dissertation, University of Waterloo, Waterloo, ON, Canada, 2010.

\end{thebibliography}
\end{document}